\documentclass[12pt]{article}

\usepackage[latin1]{inputenc} 
\usepackage{amssymb} 
\usepackage{amsmath}
\usepackage{amsthm}
\usepackage{mathtools}
\usepackage{txfonts}
\usepackage{latexsym}
\usepackage{graphicx}
\usepackage[normalem]{ulem}
\usepackage{amsfonts}
\usepackage[usenames,dvipsnames]{xcolor}
\usepackage{verbatim}
\usepackage[usenames,dvipsnames]{xcolor}

\newtheorem{defn}{Definition}[section]
\newtheorem{lem}{Lemma}[section]
\newtheorem{remark}{Remark}[section]

\newtheorem{thm}{Theorem}[section]

\newtheorem{prop}{Proposition}[section]

\numberwithin{equation}{section}
\numberwithin{figure}{section}

\newcommand{\beas}{\begin{eqnarray*}}
\newcommand{\eeas}{\end{eqnarray*}}
\newcommand{\bea}{\begin{eqnarray}}
\newcommand{\eea}{\end{eqnarray}}
\newcommand{\ben}{\begin{enumerate}}
\newcommand{\een}{\end{enumerate}}
\newcommand{\bi}{\begin{itemize}}
\newcommand{\ei}{\end{itemize}}
\newcommand{\beq}{\begin{equation}}
\newcommand{\eeq}{\end{equation}}
\newcommand{\bv}{\begin{verbatim}}
\newcommand{\ev}{\end{verbatim}}

\newcommand{\E}{\mathbb{E}}

\newcommand{\df}{\mathrm{d}}

\date{May 11, 2021}
\begin{document}

\title{\bf A rough SABR formula}

\author{ 
Masaaki Fukasawa\\
 {\small Graduate School of Engineering Science, Osaka University}\\
 {\small 1-3 Machikaneyama, Toyonaka, Osaka, JAPAN}\\
 {\small fukasawa@sigmath.es.osaka-u.ac.jp }\\~\\
 
Jim Gatheral\\{\small Baruch College, City University of New York}\\
{\small jim.gatheral@baruch.cuny.edu}\\
 
}

\maketitle\thispagestyle{empty}

\begin{abstract}
Following an approach originally suggested by Balland in the context of
 the SABR model, we derive an ODE that is satisfied by 
 normalized volatility smiles for short maturities under a rough volatility
 extension of the SABR model that extends also the rough Bergomi model.  We solve
 this ODE numerically and further present a very accurate approximation
 to the numerical solution that we dub the {\em rough SABR formula}.  
\end{abstract}

\section{Introduction}
The now twenty years old SABR (stochastic alpha-beta-rho) model remains
very popular amongst practitioners, particularly those in foreign
exchange and interest rate derivative markets.  One key reason for the
popularity of the SABR model is the existence of a 
closed-form approximation to the implied volatility smile, the
celebrated SABR formula of Hagan et al. \cite{hagan2002managing}.
SABR implied volatility surfaces are, however, not really
consistent with market data. In particular, 
the SABR implied volatility surface and its approximation by the SABR formula
cannot reproduce the power-law type term structure of at-the-money (ATM) skew
typically observed in markets.
As a result, practioners are forced to use different SABR parameters for
different maturities.
 Recently, a class of stochastic volatility models where volatility is not a
semimartingale but has a rougher path, so called {\em rough volatility} models,
have been shown to generate better fits to the observed volatility surface with very few parameters;
see~\cite{alos2007short,bayer2016pricing,fukasawa2017short,forde2017asymptotics}
for more details.  In particular, the  rough Bergomi model of
\cite{bayer2016pricing} has only three parameters.  
In this article, we introduce a rough SABR model that includes both the
SABR model and the rough Bergomi model as particular cases,
and extend the SABR formula to the rough SABR model,
demonstrating the accuracy of our formula using numerical simulation.

The SABR formula is a small-time asymptotic approximation, and 
there are several approaches to its derivation, including~\cite{berestycki2004computing,osajima2007asymptotic}.
Balland's derivation~\cite{balland2006forward} of the lognormal SABR formula 
is a particularly simple and elegant one.  
Balland's idea is to start with the drift condition for arbitrage-free implied
volatility processes, 
and then to derive an ordinary differential equation (ODE)
to be satisfied by normalized implied volatility smiles.
In this article, we apply this idea to the rough SABR model
to derive a modified ODE, which we solve numerically in the general case.  We also provide
an accurate closed-form approximation to this numerical solution, which
we dub the {\em rough SABR formula}.
In contrast to the classical SABR formula,
the rough SABR formula generates a reasonable shape for the entire volatility surface, not just a single smile.

\section{No arbitrage dynamics of the implied volatility}
As a preliminary step, here we describe the drift condition for
arbitrage-free dynamics of implied volatility processes.
We treat both the Black-Scholes and Bachelier (normal) implied volatilities. 
We follow Balland~\cite{balland2006forward} for the Black-Scholes case, 
and apply the same idea to the Bachelier case.
We also introduce the notion of asymptotically arbitrage-free
approximation of the implied volatility.

Let $S = \{S_t\}$ be the underlying asset price process of an option market, and $C=\{C_t(K,T)\}$,
$C_t(K,T) = P^{\mathrm{BS}}(S_t,K,T-t,\Sigma^{\mathrm{BS}}_t) = P^{\mathrm{B}}(S_t,K,T-t,\Sigma^{\mathrm{B}}_t) $ be a call option price process
with
strike $K$ and maturity $T$, where $P^{\mathrm{BS}}(S,K,\tau,\sigma)$
and
$ P^{\mathrm{B}}(S,K,\tau,\sigma)$ are respectively the
Black-Scholes and Bachelier call prices with volatility $\sigma$
\begin{equation*}
\begin{split}
& P^{\mathrm{BS}}(S,K,\tau,\sigma) = S\Phi(d_+) - K\Phi(d_-), \ \ d_\pm = 
\frac{\log \frac{S}{K}}{\sigma\sqrt{\tau}} \pm \frac{\sigma \sqrt{\tau}}{2},\\
& P^{\mathrm{B}}(S,K,\tau,\sigma) = \sigma\sqrt{\tau}\phi\left(
\frac{K-S_0}{\sigma\sqrt{\tau}} \right) - (K-S_0)\left(1-\Phi\left(\frac{K-S_0}{\sigma\sqrt{\tau}}\right)\right)
\end{split}
\end{equation*}
We fix $K>0$ and $T>0$, and study the Black-Scholes and Bachelier implied volatility processes
$\Sigma^{\mathrm{BS}} = \{\Sigma^{\mathrm{BS}}_t\}$ and
$\Sigma^{\mathrm{B}} = \{\Sigma^{\mathrm{B}}_t\}$.

\begin{prop}\label{prop1}
Assume that  $S$ and $C$ are continuous It\^o processes,
 and that there is an equivalent local martingale measure $Q$ for
$S$ and $C$.
Then, $\Sigma := \Sigma^{\mathrm{BS}}$ is a continuous It\^o process and, 
denoting by $D\,\mathrm{d}t$ the drift part of
 $\mathrm{d}\Sigma$ under $Q$, we have
\begin{equation}\label{na}
\begin{split}
 &   \frac{\mathrm{d}}{\mathrm{d}t}
 \langle \log S \rangle
+ 2 k
\frac{\mathrm{d}}{\mathrm{d}t}
\langle \log S, \log \Sigma \rangle
+ k^2
\frac{\mathrm{d}}{\mathrm{d}t}
\langle \log \Sigma \rangle
\\ &= \Sigma^2 - 2\Sigma\tau D
-\Sigma^2\tau
\frac{\mathrm{d}}{\mathrm{d}t}
\langle \log S, \log \Sigma \rangle
+ \frac{\Sigma^4\tau^2}{4}\frac{\mathrm{d}}{\mathrm{d}t}
\langle \log \Sigma \rangle,
\end{split}
\end{equation}
where $\tau = T-t$ and  $k = \log K/S$.
\end{prop}

\begin{proof}
By It\^o's formula,
\begin{equation*}
\mathrm{d}C-  P_S \mathrm{d}S  =
- P_\tau\mathrm{d}t  + P_\sigma\mathrm{d}\Sigma
+ \frac{1}{2}P_{SS}\mathrm{d}\langle S\rangle 
+ P_{S\sigma}\mathrm{d}\langle S, \Sigma \rangle +
\frac{1}{2}P_{\sigma\sigma}\mathrm{d}\langle \Sigma \rangle
\end{equation*}
and this is (the differential of) a local martingale under $Q$, where the greeks are
\begin{equation*}
\begin{split}
& P_S= \Phi(d_+),\ \ 
 P_\tau = S\phi(d_+)\frac{\sigma}{2\sqrt{\tau}},\ \
P_{\sigma} = S\phi(d_+)\sqrt{\tau},\\
&P_{SS} = \frac{\phi(d_+)}{S\sigma \sqrt{\tau}},\ \
P_{S\sigma} = - \frac{\phi(d_+)d_-}{\sigma},\ \
P_{\sigma\sigma} = \frac{S\sqrt{\tau}}{\sigma}d_+d_-\phi(d_+).
\end{split}
\end{equation*}
Substituting these, we obtain
\begin{equation*}
\begin{split}
 -\frac{\Sigma^2}{2}\mathrm{d}t
+ \Sigma \tau D \mathrm{d}t
&+ \frac{1}{2}\mathrm{d}\langle \log S \rangle - 
\left(\log\frac{S}{K} - \frac{\Sigma^2\tau}{2}\right)
\mathrm{d}\langle \log S, \log \Sigma \rangle\\
&+ \frac{1}{2}\left(
\left|\log \frac{S}{K}\right|^2 - \frac{\Sigma^4\tau^2}{4}\right)\mathrm{d}
\langle \log \Sigma \rangle = 0.
\end{split}
\end{equation*}

\end{proof}

In light of Proposition~\ref{prop1},
we introduce the following notion.
\begin{defn}
 A continuous It\^o process $\hat{\Sigma} = \{\hat\Sigma_t\}$ is said to
 be an asymptotically arbitrage-free approximation of $\Sigma^\mathrm{BS}$ under $Q$ if,
denoting by $\hat{D}\,\mathrm{d}t$ the drift part of
 $\mathrm{d}\hat\Sigma$ under $Q$, there exist
a continuous function
 $\varphi$ on $\mathbb{R}$
and a continuous process $\Psi = \{\Psi_t\}$ on $[0,T]$ such that
\begin{equation}\label{ana}
\begin{split}
 & \Biggl|  \frac{\mathrm{d}}{\mathrm{d}t}
 \langle \log S \rangle
+ 2 k
\frac{\mathrm{d}}{\mathrm{d}t}
\langle \log S, \log \hat\Sigma \rangle
+ k^2
\frac{\mathrm{d}}{\mathrm{d}t}
\langle \log \hat\Sigma \rangle
\\ &- \hat\Sigma^2 + 2\hat\Sigma\tau \hat{D}
+\hat\Sigma^2\tau
\frac{\mathrm{d}}{\mathrm{d}t}
\langle \log S, \log \hat\Sigma \rangle
- \frac{\hat\Sigma^4\tau^2}{4}\frac{\mathrm{d}}{\mathrm{d}t}
\langle \log \hat\Sigma \rangle \Biggr| \leq \varphi(\Psi \hat{\Sigma})\cdot
 o_p(1)
\end{split}
\end{equation}
as $\tau = T-t \to 0$, where $k = \log K/S$ and 
$o_p(1)$ is a term which converges to $0$
 in probability.
\end{defn}

Now we give a Bachelier version.

\begin{prop}\label{prop2}
Assume that  $S$ and $C$ are continuous It\^o processes,
 and that there is an equivalent local martingale measure $Q$ for
$S$ and $C$.
Then, $\Sigma := \Sigma^{\mathrm{B}}$ is a continuous It\^o process and, 
denoting by $D\,\mathrm{d}t$ the drift part of
 $\mathrm{d}\Sigma$ under $Q$, we have
\begin{equation}\label{eq:na2}
 \frac{\mathrm{d}}{\mathrm{d}t}
 \langle S \rangle
+ 2 k
\frac{\mathrm{d}}{\mathrm{d}t}
\langle S, \log \Sigma \rangle
+ k^2
\frac{\mathrm{d}}{\mathrm{d}t}
\langle \log \Sigma \rangle = \Sigma^2 - 2\Sigma\tau D,
\end{equation}
where $\tau = T-t$ and  $k = K-S$.
\end{prop}

\begin{proof}
By It\^o's formula,
\begin{equation*}
\mathrm{d}C-  P_S \mathrm{d}S  =
- P_\tau\mathrm{d}t  + P_\sigma\mathrm{d}\Sigma
+ \frac{1}{2}P_{SS}\mathrm{d}\langle S\rangle 
+ P_{S\sigma}\mathrm{d}\langle S, \Sigma \rangle +
\frac{1}{2}P_{\sigma\sigma}\mathrm{d}\langle \Sigma \rangle
\end{equation*}
and this is (the differential of) a local martingale under $Q$, where
 the (Bachelier) greeks are
\begin{equation*}
\begin{split}
& P_S= 1-\Phi\left(\frac{k}{\sigma\sqrt{\tau}}\right), \ \ 
 P_\tau =
 \frac{\sigma}{2\sqrt{\tau}}\phi\left(\frac{k}{\sigma\sqrt{\tau}}\right),
 \ \ 
P_{\sigma} = \sqrt{\tau} \phi\left(\frac{k}{\sigma\sqrt{\tau}}\right),\\
&P_{SS} = \frac{1}{\sigma \sqrt{\tau}}\phi\left(\frac{k}{\sigma\sqrt{\tau}}\right),\ \
P_{S\sigma} =  \frac{k}{\sigma^2\sqrt{\tau}}\phi\left(\frac{k}{\sigma\sqrt{\tau}}\right),\ \
P_{\sigma\sigma} = \frac{k^2}{\sigma^3\sqrt{\tau}}\phi\left(\frac{k}{\sigma\sqrt{\tau}}\right).
\end{split}
\end{equation*}
Substituting these, we obtain Equation \eqref{eq:na2}.
\end{proof}

In light of Proposition~\ref{prop2},
we introduce the following notion.
\begin{defn}
 A continuous It\^o process $\hat{\Sigma} = \{\hat\Sigma_t\}$ is said to
 be an asymptotically arbitrage-free approximation of
 $\Sigma^\mathrm{B}$ under $Q$ if,
denoting by $\hat{D}\,\mathrm{d}t$ the drift part of
 $\mathrm{d}\hat\Sigma$ under $Q$, there exist
a continuous function
 $\varphi$ on $\mathbb{R}$
and a continuous process $\Psi = \{\Psi_t\}$ on $[0,T]$ such that
\begin{equation}\label{ana2}
 \Biggl|  \frac{\mathrm{d}}{\mathrm{d}t}
 \langle S \rangle
+ 2 k
\frac{\mathrm{d}}{\mathrm{d}t}
\langle  S, \log \hat\Sigma \rangle
+ k^2
\frac{\mathrm{d}}{\mathrm{d}t}
\langle \log \hat\Sigma \rangle - \hat\Sigma^2 + 2\hat\Sigma\tau \hat{D}
 \Biggr| \leq \varphi(\Psi \hat{\Sigma})\cdot
 o_p(1)
\end{equation}
as $\tau = T-t \to 0$, where $k=K-S$ and 
$o_p(1)$ is a term which converges to $0$
in probability.
\end{defn}

It should be noted that the condition for an approximation to be an asymptotically arbitrage-free is
necessary but not sufficient for the approximation to be
reasonable. In particular, an asymptotically arbitrage-free
approximation is not unique and is not necessarily accurate.
Still, under the lognormal SABR model
\begin{equation}\label{cSABR}
 \frac{\mathrm{d}S}{S} = \alpha \mathrm{d}Z, \ \ 
\frac{\mathrm{d}\alpha}{\alpha} =  \frac{\eta}{2} \mathrm{d}W, 
\end{equation}
where $(Z,W)$ is a 2-dim correlated Brownian motion 
with $\mathrm{d}\langle Z,W \rangle_t =
\rho \mathrm{d}t$, and $ \rho \in (-1,1) $ and $\eta> 0$ are constants,
 Balland~\cite{balland2006forward} found that
\begin{equation*}
 \hat{\Sigma} : = \alpha f(Y), \ \ Y  = \frac{\eta}{\alpha} \log \frac{K}{S}
\end{equation*}
is (in our terminology) an asymptotically arbitrage-free approximation of
$\Sigma^{\mathrm{BS}}$ if $f$ is a solution of the ODE
\begin{equation*}
 \left(1 - y \frac{f^\prime(y)}{f(y)}\right)^2
\left(1 + \rho y + \frac{y^2}{4}  \right)
= f(y)^2,
\end{equation*}
and that solving this ODE, the lognormal SABR formula of Hagan et
al. \cite{hagan2002managing} is obtained:
\begin{equation*}
 f(y) = \frac{y}{g(y)}, \ \ 
g(y) = -2 \log \frac{\sqrt{1 + \rho y + y^2/4}-\rho-y/2}{1-\rho}.
\end{equation*}
See also \cite{andreasen2013expanded} for a related work.
This simple and elegant approach to reach this accurate formula
motivates us to seek an asymptotically arbitrage-free approximation to
derive a useful formula under rough volatility models.

\section{Implied volatility under rough SABR}

Here we present our rough SABR model with an asymptotically
arbitrage-free approximation of the
implied volatility. The model is
\begin{equation}
\frac{\mathrm{d}S_t}{\beta(S_t)} =  \alpha_t\,\mathrm{d}Z_t, \ \
\frac{\mathrm{d}\xi_t(s)}{\xi_t(s)} = \kappa(s-t)\mathrm{d}W_t, \ \ t < s
\label{eq:roughSABR}
\end{equation}
under an equivalent martingale measure $Q$,
where $\alpha_t = \sqrt{\xi_t(t)}$, $\beta$ is a positive continuous function,
$(Z,W)$ is a 2-dim correlated $\{\mathcal{F}_t\}$-Brownian
motion with $\mathrm{d}\langle Z,W \rangle_t = \rho
\mathrm{d}t$, $\kappa(t) = \eta \sqrt{2H}t^{H-1/2}$, $\rho \in [-1,1]$, $\eta
> 0$ and $H \in (0,1/2]$. 
We assume $\{\xi_0(s)\}_{s \geq 0}$ to be a family of $\mathcal{F}_0$
measurable random variables and the curve $s \mapsto \xi_0(s)$ to be
continuous.

Note that we have an explicit expression
\begin{equation}\label{exex}
 \xi_t(s)  = \E^Q[\xi_s(s)|\mathcal{F}_t]
= \xi_0(s) \exp\left\{ \eta \sqrt{2H}
\int_0^t (s-u)^{H-1/2}\mathrm{d}W_u - \frac{1}{2}\eta^2 (s^{2H} -(s-t)^{2H})
\right\}
\end{equation}
for $0 \leq t \leq s$. The case $H = 1/2$ with
\begin{equation*}
 \xi_0(s) = \alpha_0^2 \exp\left\{\tfrac{1}{4}\eta^2 s\right\}
\end{equation*}
is the classical SABR
model; see (\ref{cSABR}) for the log normal case $(\beta(s) = s)$. 
When $\beta(s) = s$, this is the rough Bergomi model introduced in
\cite{bayer2016pricing}, and $s \mapsto \xi_t(s)$ is the forward
variance curve at time $t$:
\begin{equation*}
\int_t^s \xi_t(u)\mathrm{d}s = \int_t^s
 \E^Q\left[\alpha_u^2|\mathcal{F}_t\right]\mathrm{d}u
= \E^Q\left[
\int_t^s \mathrm{d}\langle \log S \rangle | \mathcal{F}_t
\right].
\end{equation*}
In general, $\xi$ can be determined from weighted variance swap rates:
\begin{equation*}
\int_t^s \xi_t(u)\mathrm{d}u = \int_t^s
 \E^Q\left[\alpha_u^2|\mathcal{F}_t\right]\mathrm{d}u 
= \E^Q\left[
\int_t^s \frac{S_u^2}{\beta(S_u)^2}\mathrm{d}\langle \log S \rangle_u | \mathcal{F}_t
\right].
\end{equation*}
See \cite{fukasawa2021hedging} for the infinite dimensional Markov
property of this model with application to hedging.

Extending Balland~\cite{balland2006forward}, we obtain the following result.
\begin{thm}\label{main}
 Let $f$ be a  solution of the ODE
\begin{equation}\label{odef2}
 \left(1 - y \frac{f^\prime(y)}{f(y)}\right)^2
\left(1 + 2\rho \frac{y}{2H+1} + \left(\frac{y}{2H+1}\right)^2  \right)
= f(y)^2\left(1- (1-2H)\frac{yf^\prime(y)}{f(y)}\right)
\end{equation}
with $f(0) = 1$. Let $\beta(s) =s$, that is, consider the rough Bergomi
 model. Then,
\begin{equation*}
 \hat\Sigma := U  f(Y)
\end{equation*}
is an asymptotically arbitrage-free approximation of $\Sigma^{\mathrm{BS}}$ under $Q$, where
\begin{equation*}
 U_t = \sqrt{\frac{1}{T-t}\int_t^T \xi_t(s) \mathrm{d}s}, \ \ 
 Y_t = \frac{\kappa(T-t)}{U_t}\log \frac{K}{S_t}.
\end{equation*}
\end{thm}
The proof of Theorem~\ref{main} is given in Appendix \ref{sec:31proof}.

\begin{remark}[Consistency with the asymptotic skew formula of \cite{euch2019short}]\upshape
Substituting a formal series expansion $f(x) = 1 + ax + bx^2/2 + \dots$ to
(\ref{odef2}),
we find
\begin{equation*}
 a = \frac{\rho}{2(H+1/2)(H+3/2)},
\end{equation*}
consistent with the asymptotic skew formula for the rough Bergomi model 
 we obtained
 in~\cite{euch2019short}.  That said, the beta function coefficient that
 appears in the curvature formula of~\cite{euch2019short} precludes
 consistency with the second order coefficient
\begin{equation*}
 b = \frac{(2H+3)^2-12(2H+1)\rho^2}{2(H+1)(2H+1)^2(2H+3)^2}.
\end{equation*} 
\end{remark}

Replacing the role of $\log S$ by $S$ in the proof of Theorem~\ref{main},
we obtain the following Bachelier version.

\begin{thm}\label{main2}
 Let $f$ be a  solution of the ODE (\ref{odef2}). 
Let $\beta(s) =1$. Then,
\begin{equation*}
 \hat\Sigma := U  f(Y)
\end{equation*}
is an asymptotically arbitrage-free approximation of $\Sigma^{\mathrm{B}}$ under $Q$, where
\begin{equation*}
 U_t = \sqrt{\frac{1}{T-t}\int_t^T \xi_t(s) \mathrm{d}s}, \ \ 
 Y_t = \kappa(T-t)\frac{K-S_t}{U_t}.
\end{equation*}
\end{thm}

According to \cite{berestycki2002asymptotics}, for a local volatility model
$
\mathrm{d}S_t= \sigma(S_t)\,\mathrm{d}Z_t
$,
we have
\begin{equation*}
 \Sigma^{\mathrm{BS}} \approx \frac{\log \frac{K}{S}}{\int_S^K
  \frac{\mathrm{d}s}{\sigma(s)}}.
\end{equation*}
In particular for the Bachelier model $\sigma(s) = \sigma$, we have
\begin{equation*}
 \Sigma^{\mathrm{BS}} \approx \sigma \frac{\log \frac{K}{S}}{K-S}
\end{equation*}
that connects the Black-Scholes and Bachelier volatility parameters.
Combining this and Theorem~\ref{main2}, we obtain an approximation
formula for our model with $\beta(s) = 1$: 
\begin{equation*}
 \Sigma^{\mathrm{BS}}\approx
Uf(Y) \frac{\log \frac{K}{S}}{K-S}
= \frac{U}{K-S} f\left(
\kappa(\tau)\frac{K-S}{U}
\right) \log \frac{K}{S}.
\end{equation*}
This further suggests a formula for general $\beta$:
\begin{equation}\label{gen}
 \Sigma^{\mathrm{BS}}\approx
 \frac{U}{X} f\left(
\kappa(\tau)\frac{X}{U}
\right) \log \frac{K}{S}, \ \ X = \int_S^K \frac{\mathrm{d}s}{\beta(s)}.
\end{equation}

\section{Solving the ODE}
Now we study the solution of the ODE (\ref{odef2}).
For $g(y) = y/f(y)$, from (\ref{odef2}), we have
\begin{equation}\label{ode}
 g^\prime(y)^2\left(1 + 2\rho \frac{y}{2H+1} + \frac{y^2}{(2H+1)^2}\right)
= 1 - (1-2H)\left(1 - \frac{yg^\prime(y)}{g(y)}\right)
\end{equation}
with $g(y)/y \to 1$ as $y \to 0$.
When, $H = 1/2$, this is solvable and 
\begin{equation}
g(y) = -2 \log \frac{\sqrt{1 + \rho y + y^2/4}-\rho-y/2}{1-\rho}
 \label{eq:H12}
\end{equation}
which gives the familiar SABR formula. For $H=0$, we also have the explicit
solution
\begin{equation}
 g(y) = \frac{y}{|y|}
\sqrt{\log(1 + 2\rho y + y^2) +\frac{2\rho}{\sqrt{1-\rho^2}} \left(\arctan\frac{\rho}{\sqrt{1-\rho^2}} - \arctan\frac{y+\rho}{\sqrt{1-\rho^2}}\right)}.
\label{eq:H0}
\end{equation}
We plot the resulting solutions $f(y) = y/g(y)$ for various values of $\rho$ in Figure~\ref{fig:41}.
\begin{figure}[tbh!]
\includegraphics[width=\linewidth]{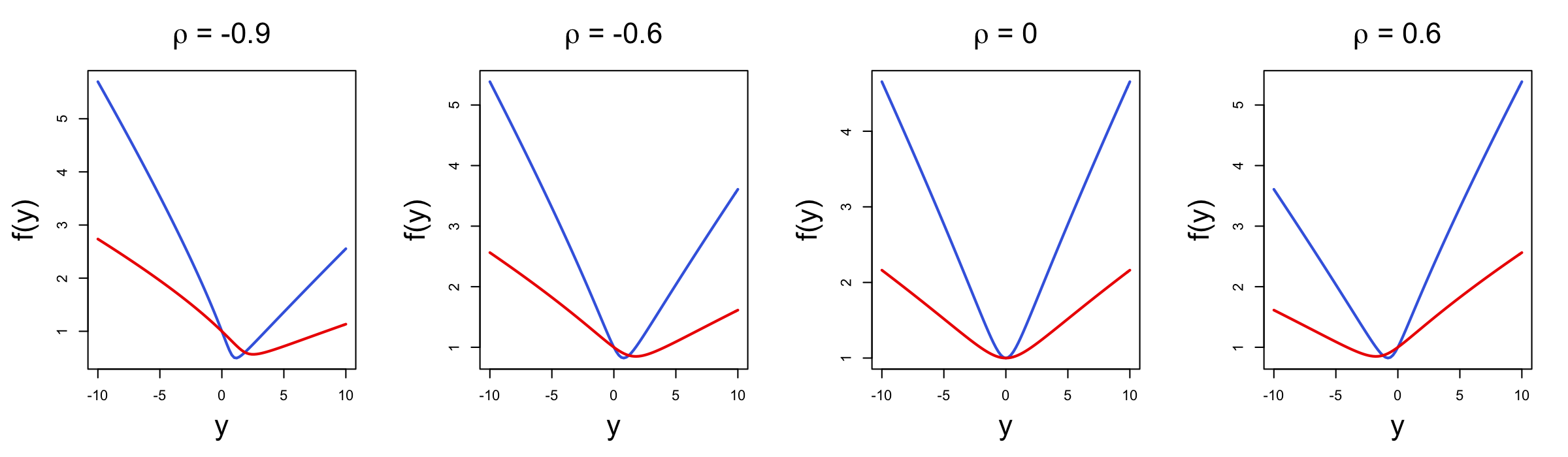}
\caption{The function $f$ for $H=1/2$ (in red) and $H=0$ (in blue).}\label{fig:41}
\end{figure}

For general $H$, the ODE has to be solved
numerically.
By (\ref{ode}), we have
\begin{equation*}
 g^\prime(y) = 
\frac{(1-2H)f(y)+ \sqrt{(1-2H)^2f(y)^2 + 8Hq(y)}}{2q(y)}, \ \ g(0) = 0
\end{equation*}
where
\begin{equation*}
f(y) = \frac{y}{g(y)}, \ \ f(0) = 1, \ \ 
 q(y) = 1 + 2\rho \frac{y}{2H+1} + \frac{y^2}{(2H+1)^2}.
\end{equation*}
There is no difficulty in obtaining a numerical solution for such a
one dimensional first order ODE.
See Figure~\ref{fig:42} for numerical examples of $f$ functions 
for various values of $H$ when $\rho
= 0$ (right) and $\rho = -0.9$ (left).

\begin{figure}
\includegraphics[width=\linewidth]{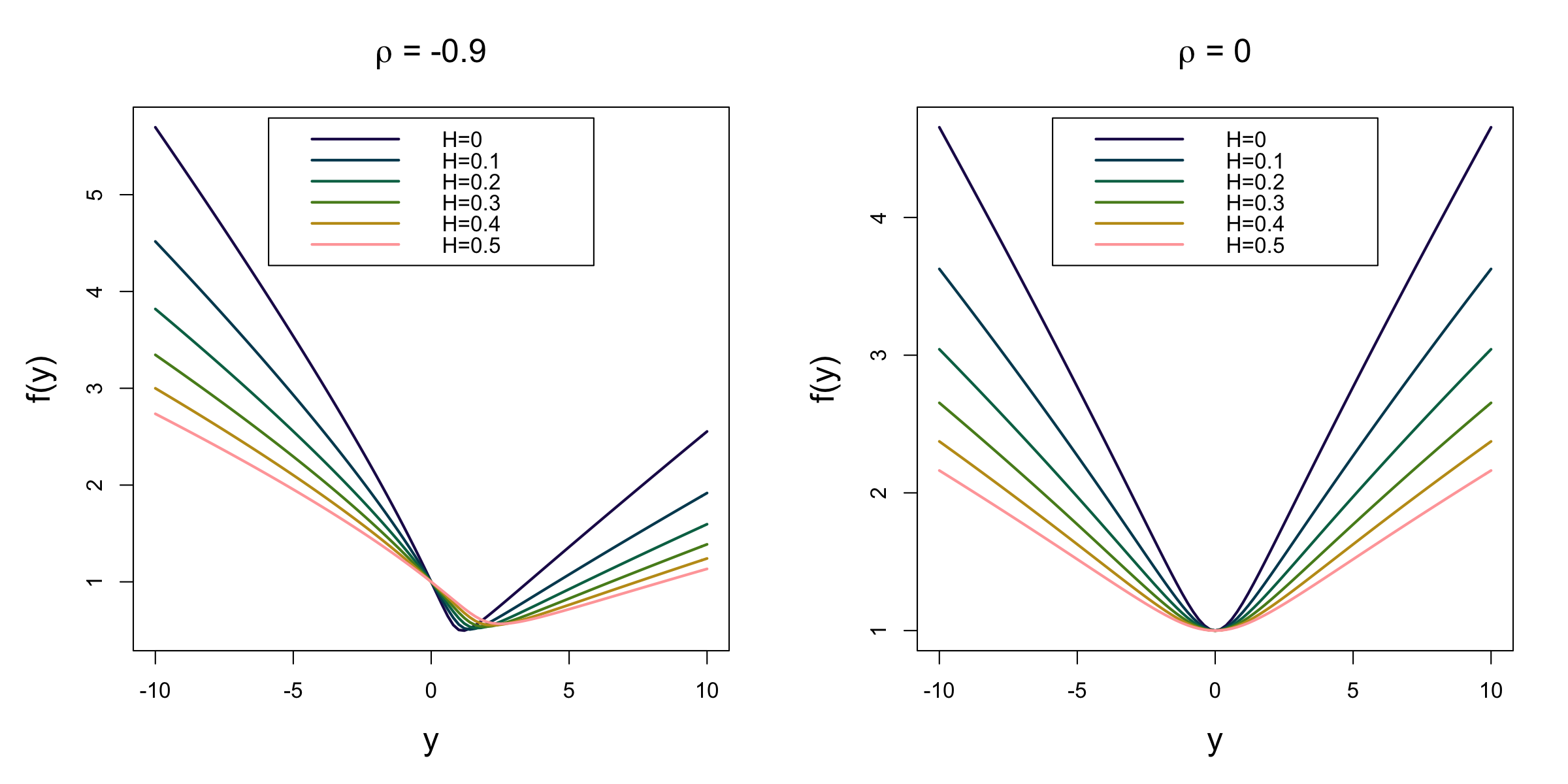}
\caption{The function $f$: numerical solutions for various values of $H$}\label{fig:42}
\end{figure}

  \section{A closed-form approximation}
  
In order to derive a closed-form approximation, it turns out to be convenient to recast \eqref{ode} in terms of $G(y) := g(y)^2$.  We find
\beq
\frac14\,G'(y)^2 \,\left\{\frac{y^2}{(2 H+1)^2} +\frac{2\,\rho\,y}{2 H+1} +1\right\} = (1-2 H)\,y\,\tfrac12 G'(y) + 2 H\,G(y).
\label{eq:bigG}
\eeq
with initial condition $G(0)=0$. Denote the solutions of \eqref{eq:bigG} with $H=0$ and $H=1/2$ respectively by $G_0(\cdot)$ and $G_{1/2}(\cdot)$.  Then from \eqref{eq:H12} and \eqref{eq:H0}, 
\beas
G_{1/2}(y) &=& 4\,\left( \log \frac{\sqrt{1 + \rho y + y^2/4}-\rho-y/2}{1-\rho} \right)^2\\
G_0(y) &=&\log(1 + 2\rho y + y^2) +\frac{2\rho}{\sqrt{1-\rho^2}} \left(\arctan\frac{\rho}{\sqrt{1-\rho^2}} - \arctan\frac{y+\rho}{\sqrt{1-\rho^2}}\right).\nonumber\\
\eeas

Substituting 
\begin{equation*}
 G(y) = y^2 + a\,y^3 + b\,y^4 + \dots
\end{equation*}
into \eqref{eq:bigG} and matching coefficients gives
\beas
 a = - \frac{\rho}{\gamma (\gamma+1)}, \ \ 
b= \frac{1}{{4 \gamma ^2\,(2\gamma+1)}}\left(3\,\rho^2\,\frac{ 4 \gamma +1 }{(\gamma +1)^2}-1\right).
\eeas
where $\gamma = H+\tfrac12$.
To the same order, we have 
\beas
G_{1/2}(y) &=&y^2-\frac{\rho  y^3}{2}+\frac{1}{48} \left(15 \rho ^2-4\right) y^4+\dots \\
G_0(y) &=&y^2 -\frac{4 \rho  y^3}{3}+  \frac{1}{2} \left(4 \rho ^2-1\right) y^4+ \dots \,.
\eeas
Matching coefficients of $y^2$ and $y^3$, we arrive at the following interpolation of the extreme solutions $G_0$ and $G_{1/2}$, which by construction gives the correct ATM skew:
\beq
 {G}_A(y) = (2H+1)^2
\left\{
\frac{3\,(1-2H)}{2H+3}\, G_0\left(\frac{y}{2H+1}\right)
+ \frac{2H}{2H+3}\,G_{1/2}\left(\frac{2y}{2H+1}\right)
\right\}.
\label{eq:GA}
\eeq

Obviously, the approximate solution $G_A(y)$ agrees with $G_0(y)$ when
$H=0$ and with $G_{1/2}(y)$ when $H = 1/2$.  To give a sense for the
accuracy of the approximate solution is for general $H$, in Figure
\ref{fig:51} we plot the numerical solution $f$ versus the approximation
$f_A(y) := |y|/\sqrt{G_A(y)}$ for two values of $H$: $H=0.05$ which is a
typical calibrated value, and $H=0.25$ which should approximately
maximize the approximation error.

\begin{figure}[tbh!]
\includegraphics[width=\linewidth]{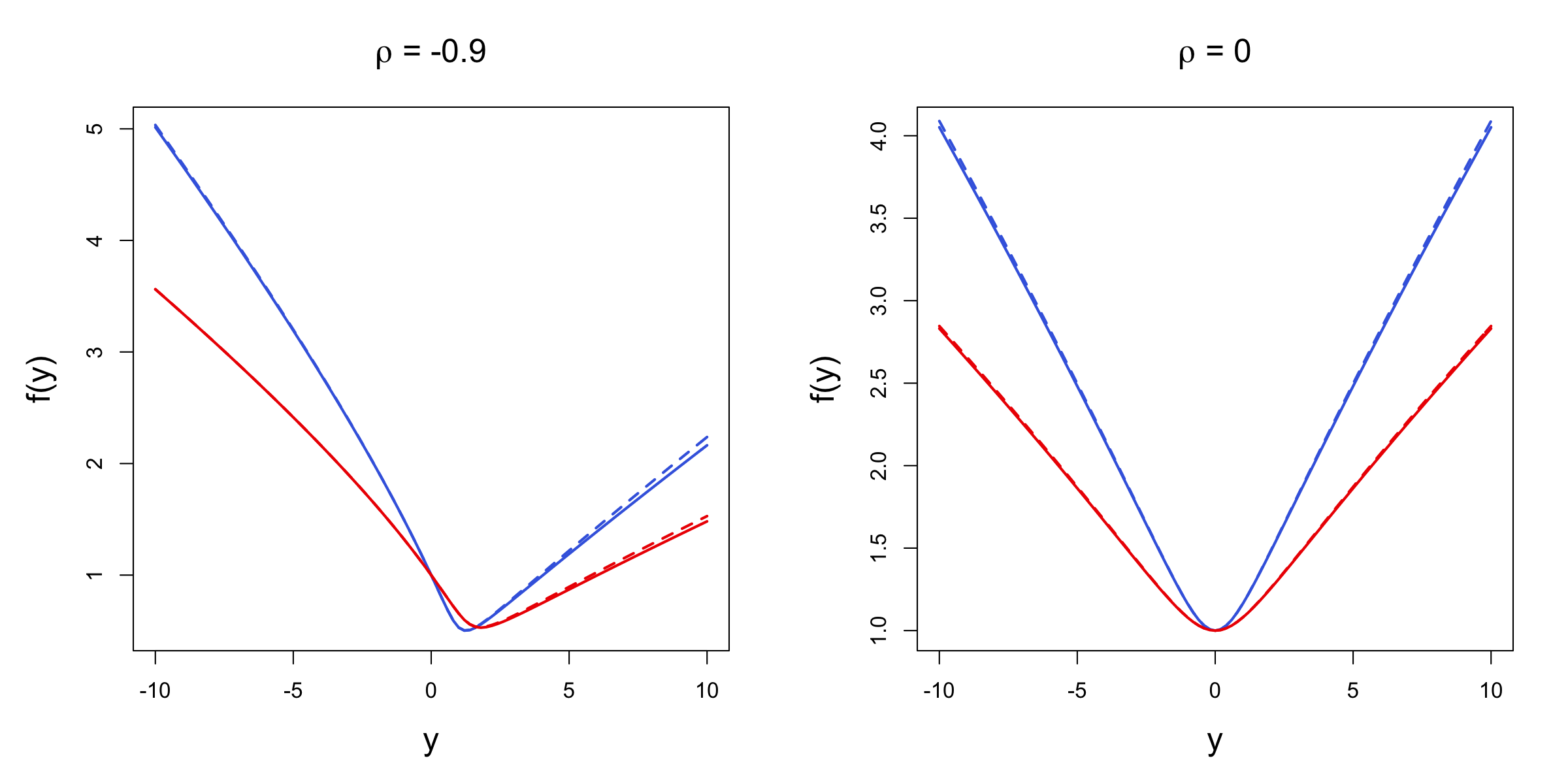}
\caption{The function $f$ and its approximation $f_A$ for two values of $\rho$.  $H=0.05$ is in blue, $H=0.25$ is in red; solid line is the numerical solution $f$ and dashed, the approximation $f_A(y)$. }\label{fig:51}
\end{figure}

\section{Final formula with Monte-Carlo comparison}
Let $\tau=T-t$ as before and
\[
y(x,\tau) = \frac{\kappa(\tau)\,x}{U(\tau)}, \ \ 
k = \log \frac{K}{S},
\ \ k_\beta = \int_S^K \frac{\mathrm{d}s}{\beta(s)},
\ \ 
U_t(\tau) = \sqrt{\frac{1}{\tau}\int_t^{t + \tau} \xi_t(s) \df s}. 
\]
Then \eqref{gen} and \eqref{eq:GA} suggest the following approximate {\em rough SABR
formula} for the Black-Scholes implied volatility $\Sigma(k,\tau)$ of an option with time to expiration
$\tau$ and log-strike $k$:

   \beq
   \boxed{
 \Sigma(k,\tau) =
   \Sigma(0,\tau)\,\frac{|y(k,\tau)|}{\sqrt{G_A(y(k_\beta,\tau))}}}
\label{eq:rSABRformula}
\eeq
under the rough SABR model \eqref{eq:roughSABR}.

 In order to confirm the accuracy of \eqref{eq:rSABRformula} in
 the lognormal case $\beta(s) = s$, we simulate the rough Bergomi model \cite{bayer2016pricing}
 with kernel $\kappa(\tau) = \eta\sqrt{2H}\,\tau^{H-1/2}$ using the
 hybrid scheme
 \cite{bennedsen2017hybrid,fukasawa2021refinement}\footnote{Implementation
 of the refinement of \cite{fukasawa2021refinement} seems to make a
 significant difference in resolving the smile for out-of-the-money
 calls. Here we set the parameters $\kappa = 2$ and $\kappa^\prime = n$ for the scheme of 
\cite{fukasawa2021refinement}.} with $1$ million paths, $2^{10}=1,024$ time steps for $H \in \{0.10,0.20\}$ and $2^{13}=8,192$ time steps for the case $H=0.05$.  

We take $\eta = 1$ and a flat forward variance curve, $\xi_0(s) = 0.04$.  In Figures~\ref{fig:H05}, \ref{fig:H10}, and \ref{fig:H20}, for $H \in \{0.05,0.1,0.2\}$, we plot smiles for $\rho \in \{-0.9,-0.6,0,0.6\}$.  Specifically, each subplot has the graph of
\begin{equation*}
\left(y(k,\tau),  \frac{\Sigma(k,\tau)}{\Sigma(0,\tau)}\right)=\left(\frac{\kappa(\tau)k}{0.2},  \frac{\Sigma(k,\tau)}{\Sigma(0,\tau)}\right),
\end{equation*}
In each case, the dashed red curve (``rSABR'') is the function $f$
obtained by solving the ODE \eqref{odef2} numerically.  From these plots,
we first notice that the scaling of $y(k,\tau)$ works remarkably well to
offset the maturity dependence of the normalized smile. Also, note in particular that the quality of our approximation decreases as
$H$ decreases.

\begin{figure}[p]
\includegraphics[width=\linewidth]{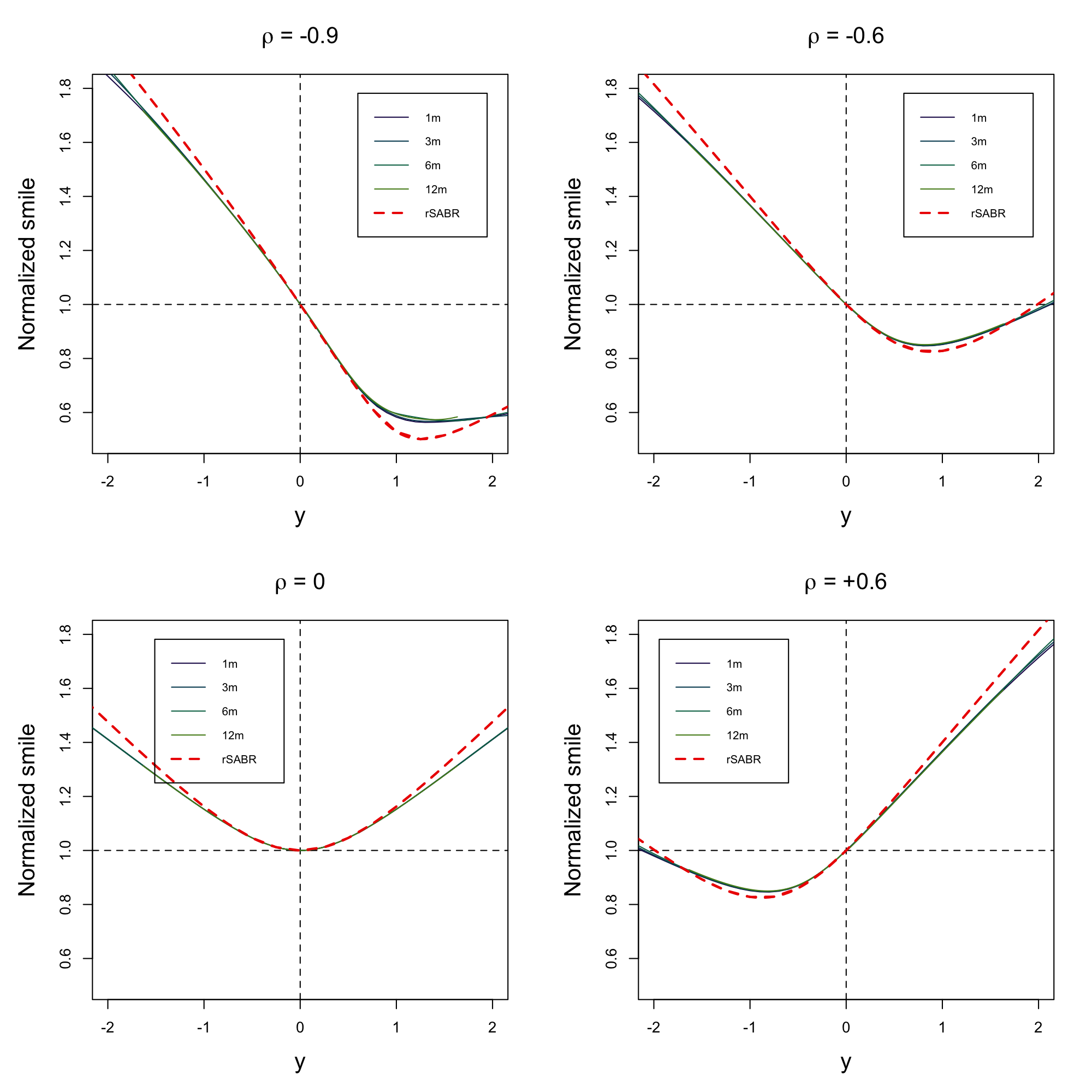}
\caption{With $\beta(s)=s$ and parameters $H=0.05$, $\eta=1$, the dashed red line is the numerical solution $f$;  Monte Carlo estimates of normalized implied volatility $\Sigma(k,\tau)/{\Sigma(0,\tau)}$ for $\tau = 1,\,3,\,6$, and $12$ months are as in the legend.}\label{fig:H05}
\end{figure}

\begin{figure}[p]
\includegraphics[width=\linewidth]{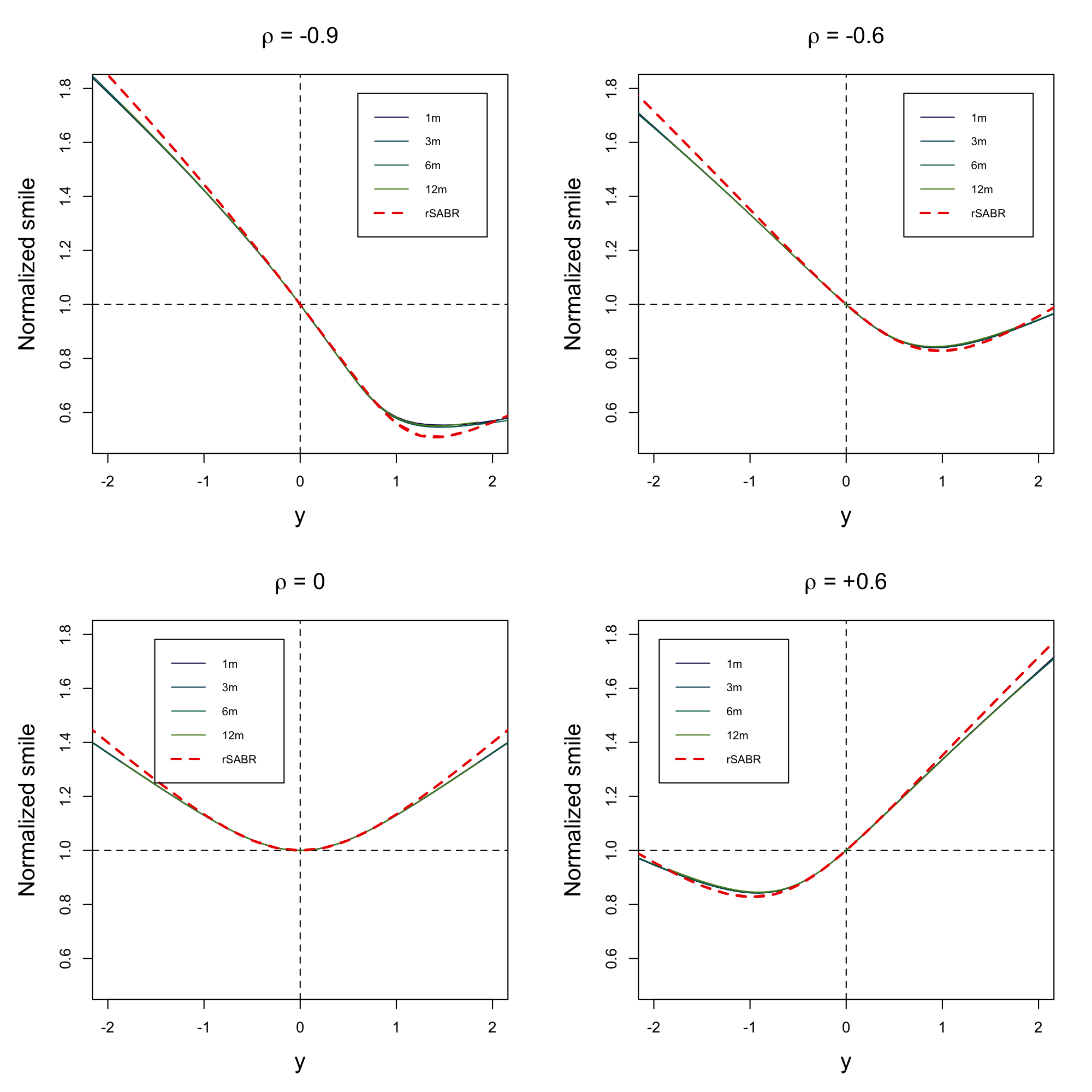}
\caption{With $\beta(s)=s$ and parameters $H=0.10$, $\eta=1$, the dashed red line is the numerical solution $f$;  Monte Carlo estimates of normalized implied volatility $\Sigma(k,\tau)/{\Sigma(0,\tau)}$ for $\tau = 1,\,3,\,6$, and $12$ months are as in the legend.}\label{fig:H10}
\end{figure}

\begin{figure}[p]
\includegraphics[width=\linewidth]{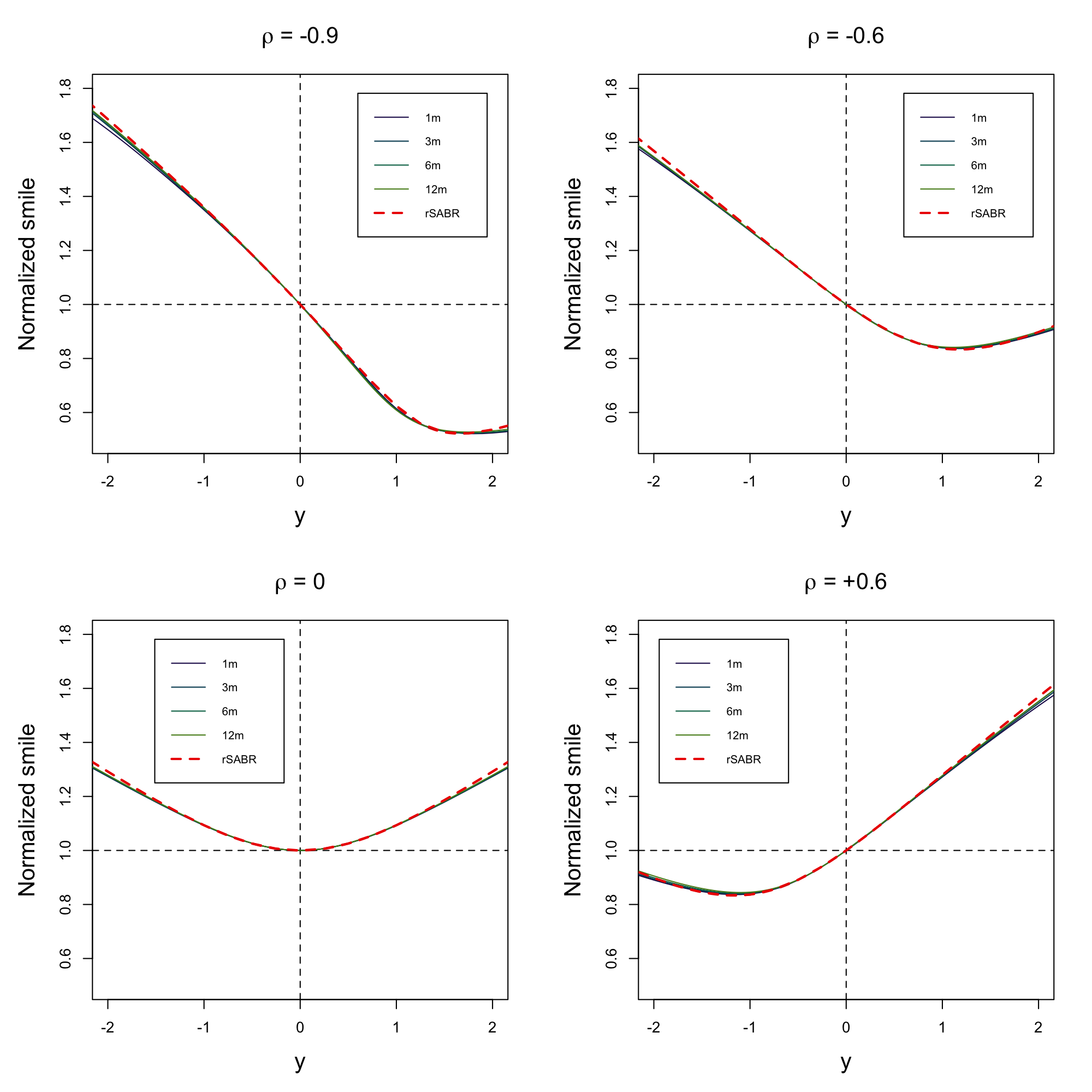}
\caption{With $\beta(s)=s$ and parameters $H=0.20$, $\eta=1$, the dashed red line is the numerical solution $f$;  Monte Carlo estimates of normalized implied volatility $\Sigma(k,\tau)/{\Sigma(0,\tau)}$ for $\tau = 1,\,3,\,6$, and $12$ months are as in the legend.}\label{fig:H20}
\end{figure}

As a further experiment, we simulate the rough SABR model
\eqref{eq:roughSABR} with $\beta(s) = \sqrt{s}$ with $1$ million paths
and $2^{12} = 4,096$ time steps for $H=0.05$ and $\rho=-0.9$.  Again, we
take $\eta = 1$ and a flat forward variance curve, $\xi_0(s) = 0.04$.
The quality of the rough SABR
formula \eqref{eq:rSABRformula}
 is demonstrated in Figure \ref{fig:64}.
Here we plot against $k$ on the x-axis, rather than $y(k,\tau)$.

Through these numerical experiments, we observe that for small
values of $H$ such as $H=0.05$ with $\eta = 1$, we need as
many as $2^{12}$ time steps to achieve convergence. A computation with so many time steps is obviously extremely time-consuming, as if to emphasize to us the value of analytical approximations.

\begin{figure}[tbh!]
\includegraphics[width=\linewidth]{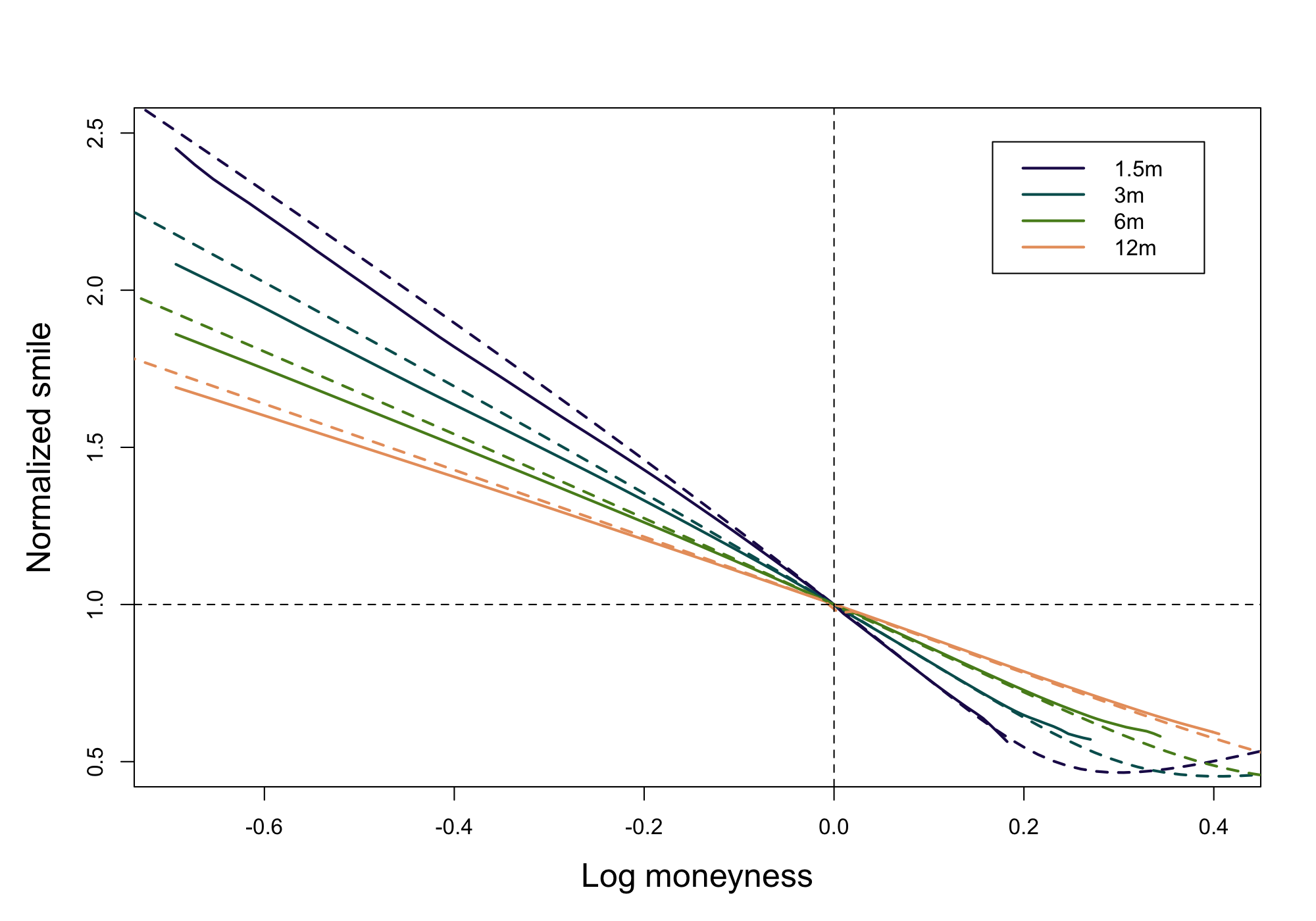}
\caption{With $\beta(s)=\sqrt{s}$ and parameters $H=0.05$, $\eta=1$, Monte Carlo estimates of normalized implied volatility $\Sigma(k,\tau)/{\Sigma(0,\tau)}$ for $\tau = 1.5,\,3,\,6$, and $12$ months are as in the legend.  Dashed lines are corresponding plots of the rough SABR formula \eqref{eq:rSABRformula}.}\label{fig:64}
\end{figure}

We emphasize that our rough SABR formula \eqref{eq:rSABRformula} is a
short-dated approximation, as is the original SABR formula of
\cite{hagan2002managing}.  As with the classical SABR formula, we can
only expect the approximation to work well for values of the
dimensionless expansion parameter $\eta\,\tau^{H} <1$ and indeed this is
what we find in our extensive numerical experiments.   For more extreme rough Bergomi parameter sets such as the pair $\eta = 2.3,\, H=0.05$ arising from calibration of the rough Bergomi model to the SPX surface \cite{bayer2016pricing}, the approximate formula \eqref{eq:rSABRformula} does not match the Monte Carlo smile sufficiently well for it to be useful in model calibration.

\section{Summary and conclusions}

Following an approach originally suggested by Balland
\cite{balland2006forward} in the context of the classical SABR model, we
derived an ODE that is satisfied by the rough SABR normalized volatility
smile for short maturities.  We solved this ODE numerically and further derived a very accurate approximation to the numerical solution.  The resulting analytical formula coincides with the classical one in the case $H = \tfrac12$.  Numerical simulation of rough Bergomi smiles confirms that our small time approximation works well for values of the expansion parameter $\eta\,\tau^{H} <1$.

In FX and interest rate applications, parameters of the classical SABR
formula are allowed to depend on time to expiration; effectively a
different model for each expiration.  In our setup, effective classical
SABR parameters also depend on time to expiration, according to the
chosen value of $H$.  Our conjecture is that the time dependence of
market-implied classical SABR parameters can be parameterized by $H$ -- confirmation of this is left for further research.  If so, given that our rough SABR formula is hardly more complicated than the classical SABR formula, it has the potential to be widely adopted by practitioners.

\bibliographystyle{alpha}
\bibliography{RoughVolatility}

\appendix

\section{Proof of Theorem~\ref{main}}\label{sec:31proof}
\begin{proof}
Denote $\tau = T-t$ and $k = \log K/S$.
By It\^o's formula, we have
\begin{equation*}
\frac{\mathrm{d}U}{U} =
 \frac{1}{2}\kappa(\tau)\,R\,\mathrm{d}W + 
\frac{1}{2}\left(
\frac{1}{\tau}\left(1 -
	    \frac{\alpha^2}{U^2}\right)-\frac{1}{4}\kappa(\tau)^2 R^2
\right) \mathrm{d}t,
\end{equation*}
where
\begin{equation*}
 R_t = \frac{\int_t^T\kappa(s-t)\xi_t(s) \mathrm{d}s}{\kappa(T-t)\int_t^T\, \xi_t(s) \mathrm{d}s}.
\end{equation*}
Further, we have
\begin{equation*}
\begin{split}
 \mathrm{d}Y &= -\frac{\kappa^\prime(\tau)}{\kappa(\tau) } Y\mathrm{d}t
+ \frac{\kappa(\tau)}{U}\mathrm{d} k 
+ \kappa(\tau)k \mathrm{d} \frac{1}{U} + \kappa(\tau)\mathrm{d}\langle
 \frac{1}{U},k \rangle
\\
& = 
\frac{Y}{2\tau} (1-2H)\mathrm{d}t + \frac{\kappa(\tau)}{U}\mathrm{d} k
- Y \frac{\mathrm{d}U}{U} + \frac{Y}{U^2} \mathrm{d}\langle U \rangle
-\frac{\kappa(\tau)}{U^2}\langle U, k\rangle
\\
&= \frac{Y}{2\tau}\left( \frac{\alpha^2}{U^2}- 2H  \right) \mathrm{d}t
+ \frac{\kappa(\tau)^2}{2}\left(
\frac{3}{4}YR^2 + \frac{\alpha}{U}R\rho
\right)\mathrm{d}t
+ \frac{\kappa(\tau)}{U}\mathrm{d} k - \frac{Y}{2}\kappa(\tau)R
 \mathrm{d}W \\
&= - \frac{Y}{2} \kappa(\tau)R \mathrm{d}W
- \kappa(\tau) \frac{\alpha }{U}\mathrm{d}Z + \text{drift},
\end{split}
\end{equation*}
and so,
\begin{equation*}
 \begin{split}
  \mathrm{d}\log \hat\Sigma &= \frac{1}{2}\kappa(\tau)R\mathrm{d}W +
  \frac{f^\prime(Y)}{f(Y)}\mathrm{d}Y + \text{drift},\\
&=\frac{1}{2}\kappa(\tau)R\left(1-\frac{Y f^\prime(Y)}{f(Y)}\right)\mathrm{d}W
-  \kappa(\tau)  \frac{f^\prime(Y)}{f(Y)}
\frac{\alpha}{U} \mathrm{d}Z+ \text{drift}.
 \end{split}
\end{equation*}
This implies 
\begin{equation*}
\begin{split}
&\frac{\mathrm{d}}{\mathrm{d}t} \langle \log \hat\Sigma \rangle
\\ &= \kappa(\tau)^2\left\{
\frac{R^2}{4}\left(1-\frac{Yf^\prime(Y)}{f(Y)}\right)^2
 + 
\left(\frac{f^\prime(Y)}{f(Y)}
\frac{\alpha }{U}\right)^2  - R
\left(1-\frac{Y f^\prime(Y)}{f(Y)}\right)\frac{f^\prime(Y)}{f(Y)}
\frac{\rho \alpha}{U}\right\}
\end{split}
\end{equation*}
and
\begin{equation*}
 \frac{\mathrm{d}}{\mathrm{d}t} \langle \log S, \log \hat\Sigma \rangle
=
\frac{\rho\alpha}{2}\kappa(\tau)R_t\left(1-\frac{Y f^\prime(Y)}{f(Y)}\right) 
- \frac{f^\prime(Y)}{f(Y)}
\frac{\alpha^2 \kappa(\tau)}{U}.
\end{equation*}
Therefore,
\begin{equation*}
 \begin{split}
&   \frac{\mathrm{d}}{\mathrm{d}t}
 \langle \log S \rangle
+ 2 k
\frac{\mathrm{d}}{\mathrm{d}t}
\langle \log S, \log \hat\Sigma \rangle
+ k^2
\frac{\mathrm{d}}{\mathrm{d}t}
\langle \log \hat\Sigma \rangle 
\\ =  & \alpha^2 + 2k \kappa(\tau)
\left\{\frac{\rho\alpha}{2}R\left(1-\frac{Yf^\prime(Y)}{f(Y)}\right) 
- \frac{f^\prime(Y)}{f(Y)}
\frac{\alpha^2}{U}\right\} \\
& +k^2\kappa(\tau)^2
\left\{
\frac{R^2}{4}\left(1-\frac{Yf^\prime(Y)}{f(Y)}\right)^2
 + 
\left(\frac{f^\prime(Y)}{f(Y)}
\frac{\alpha}{U}\right)^2  -
R\left(1-\frac{Yf^\prime(Y)}{f(Y)}\right)\frac{f^\prime(Y)}{f(Y)}
\frac{\rho \alpha}{U}
\right\} \\
= &U^2 \Biggl\{
\frac{\alpha^2}{U^2} + 2Y \frac{\alpha}{U}\left\{\frac{\rho
  R}{2}\left(1-\frac{Yf^\prime(Y)}{f(Y)}\right) 
- \frac{f^\prime(Y)}{f(Y)} \frac{\alpha}{U}
\right\} \\
& \hspace*{1cm}
+ Y^2 \left\{
\frac{R^2}{4}\left(1-\frac{Yf^\prime(Y)}{f(Y)}\right)^2
 + 
\left(\frac{f^\prime(Y)}{f(Y)}  \frac{\alpha}{U}
\right)^2  -
\left(1-\frac{Yf^\prime(Y)}{f(Y)}\right)\frac{f^\prime(Y)}{f(Y)}
  \frac{\alpha}{U}\rho R
\right\}
\Biggr\} \\
= & 
U^2\Biggl\{
\left( \frac{\alpha}{U} +
Y \left\{\frac{\rho R}{2}\left(1-\frac{Yf^\prime(Y)}{f(Y)}\right) 
- \frac{f^\prime(Y)}{f(Y)} \frac{\alpha}{U}
\right\}
\right)^2 
+ (1-\rho^2)\frac{R^2}{4}
Y^2\left(1-\frac{Yf^\prime(Y)}{f(Y)}\right)^2
\Biggr\}
\\
=& 
U^2
\left(1 -\frac{Yf^\prime(Y)}{f(Y)}\right)^2
\left\{
\frac{\alpha^2}{U^2} + \rho R \frac{\alpha}{U} Y + \frac{R^2}{4}Y^2
\right\}.
 \end{split}
\end{equation*}
Now we consider the second line of (\ref{ana}).
From the above computations, we see that the last two terms are of
 $O(\tau \kappa(\tau)) = O(\tau^{H+1/2})$. For the second term,
since
\begin{equation*}
 \mathrm{d}\hat\Sigma = 
f(Y)\mathrm{d}U + Uf^\prime(Y)\mathrm{d}Y + 
\frac{1}{2}Uf^{\prime\prime}(Y)\mathrm{d}\langle Y \rangle + 
f^\prime(Y)\mathrm{d}\langle U, Y \rangle,
\end{equation*}
we have
\begin{equation*}
 \begin{split}
  \hat{D} = & Uf(Y)\left(
\frac{1}{2\tau}\left(1- \frac{\alpha^2}{U^2}B\right) - \frac{1}{8}\kappa(\tau)^2 R^2
\right) \\
&+ Uf^\prime(Y)\left(
\frac{Y}{2\tau}\left( \frac{\alpha^2}{U^2}- 2H  \right) 
+ \frac{\kappa(\tau)^2}{2}\left(
\frac{3}{4}YR^2 + \frac{\alpha}{U}R\rho
\right)
 \right)\\
& + \frac{1}{2}Uf^{\prime\prime}(Y)\kappa(\tau)^2\left(
\frac{\alpha^2}{U^2} + \frac{\alpha}{U} 
YR \rho + \frac{Y^2R^2}{4}\right)
 - Uf^\prime(Y)\kappa(\tau)^2\left(
\frac{YR^2}{4} + \frac{\alpha}{U} \frac{R\rho}{2}
\right) \\
= & \frac{1}{2\tau}Uf(Y)
\left(
1 - \frac{\alpha^2}{U^2} + \frac{Yf^\prime(Y)}{f(Y)}\left(\frac{\alpha^2}{U^2}-2H\right)
\right) \\
& + \frac{1}{2}U\kappa(\tau)^2\left(- f(Y)
\frac{R^2}{4} \left(1 - \frac{Yf^\prime(Y)}{f(Y)}\right) +
  f^{\prime \prime}(Y)
\left(
\frac{\alpha^2}{U^2} + \rho R \frac{\alpha}{U} Y + \frac{R^2}{4}Y^2
\right)
\right).
 \end{split}
\end{equation*}
Therefore, we have a nonnegligible term
\begin{equation*}
2\hat\Sigma \tau \hat{D} \approx 
U^2f(Y)^2\left(\left(1-\frac{Yf^\prime(Y)}{f(Y)}\right)
 \left(1-\frac{\alpha^2}{U^2}\right) +(1- 2H) \frac{Yf^\prime(Y)}{f(Y)}\right).
\end{equation*}
Now, using that $f$ is a solution of (\ref{odef2}),
\begin{equation*}
\begin{split}
 &   \frac{\mathrm{d}}{\mathrm{d}t}
 \langle \log S \rangle
+ 2 k
\frac{\mathrm{d}}{\mathrm{d}t}
\langle \log S, \log \hat\Sigma \rangle
+ k^2
\frac{\mathrm{d}}{\mathrm{d}t}
\langle \log \hat\Sigma \rangle
\\ &- \hat\Sigma^2 + 2\hat\Sigma\tau \hat{D}
+\hat\Sigma^2\tau
\frac{\mathrm{d}}{\mathrm{d}t}
\langle \log S, \log \hat\Sigma \rangle
- \frac{\hat\Sigma^4\tau^2}{4}\frac{\mathrm{d}}{\mathrm{d}t}
\langle \log \hat\Sigma \rangle \\
& \approx
U^2\left(1-\frac{Yf^\prime(Y)}{f(Y)}\right)^2
\Biggl(
-1+\frac{\alpha^2}{U^2} +
\rho Y\left(\frac{\alpha}{U}R - \frac{1}{H+1/2}\right)
+ \frac{Y^2}{4} \left(R^2- \frac{1}{(H+1/2)^2}\right)\Biggr)
 \\ & \hspace*{1cm} + U^2f(Y)^2\left(1-\frac{Yf^\prime(Y)}{f(Y)}\right)
 \left(1-\frac{\alpha^2}{U^2}\right).
\end{split}
\end{equation*}
The result then follows from the lemma below, by taking $\Psi = 1/U$ so
 that $\Psi\hat{\Sigma} = f(Y)$.
\end{proof}

\begin{lem} As $\tau = T -t \to 0$,
\begin{equation*}
\frac{\alpha_t}{U_t}  \to 1, \ \  R_t \to \frac{1}{H+1/2}
\end{equation*}
in probability.
\end{lem}
\begin{proof}
By the explicit expression (\ref{exex}), we have
\begin{equation*}
\log  \frac{\xi_t(s)}{\xi_t(t)}  = \log \frac{\xi_0(s)}{\xi_0(t)} + 
\int_0^t \left[\kappa(s-u)-\kappa(t-u)\right]\mathrm{d}W_u
-\frac{1}{2}\int_0^t
\left[\kappa(s-u)^2-\kappa(t-u)^2\right]\mathrm{d}u
\end{equation*}
for $t \leq s \leq T$. Now, as $\tau = T-t \to 0$,
\begin{equation*}
 \int_0^t \left[\kappa(s-u)^2-\kappa(t-u)^2\right]\mathrm{d}u
= \eta^2\left[ s^{2H} - (s-t)^{2H} - t^{2H} \right]\to 0,
\end{equation*}
and
\beas
&& \int_0^t \left[\kappa(s-u)-\kappa(t-u)\right]^2\mathrm{d}u
\\&=& \eta^2\left[s^{2H} - (s-t)^{2H} + t^{2H} 
-4H\int_0^t (s-u)^{H-1/2}(t-u)^{H-1/2} \mathrm{d}u
)\right]
\\&=& \eta^2\left[s^{2H} - (s-t)^{2H} + t^{2H} 
-4H\,t^{2H}\int_0^1 (s/t-1 +x)^{H-1/2}x^{H-1/2} \mathrm{d}x
)\right] \to 0.
\eeas
Also, the initial forward variance curve is continuous by assumption, so $\xi_t(s)/\xi_t(t) \to 1$. Therefore,
\begin{equation*}
\frac{U_t^2}{\alpha_t^2} = \int_0^1
 \frac{\xi_t(t + \tau \theta)}{\xi_t(t)}\mathrm{d}\theta \to 1,
\end{equation*}
and
\begin{equation*}
 R_t = \frac{\int_0^1\theta^{H-1/2} \,\xi_t(t + \tau \theta)\,
  \mathrm{d}\theta}{\int_0^1 \xi_t(t + \tau \theta)\, \mathrm{d}\theta} \to
\int_0^1\theta^{H-1/2}\mathrm{d}\theta = \frac{1}{H+1/2}.
\end{equation*}
\end{proof}


\end{document}